\documentclass[11pt]{amsart}
\usepackage{a4}
\usepackage{amsmath}
\usepackage{amssymb}
\usepackage{amsthm}
\usepackage{graphicx}
\usepackage[english]{babel}

\usepackage[dvips,colorlinks=true,linkcolor=blue]{hyperref}

\newcommand{\bR}{\mathbb{R}}

\newcommand{\bP}{\mathbb{P}}
\newcommand{\mintwo}[2]{\min_{\substack{#1 \\ #2}}}

\newcommand{\D}{\displaystyle}

\newcommand{\R}{{\mathbb R}}

\newcommand{\Z}{{\mathbb Z}}
\newcommand{\N}{{\mathbb N}}

\newcommand{\tr}{{\rm tr}\,}

\newcommand{\vers}{\operatornamewithlimits{\to}}
\newcommand{\esp}{\mathbb{E}}
\newcommand{\pro}{\mathbb{P}}
\newcommand{\car}{\bold{1}}

\newtheorem{theorem}{Theorem}%[section]
\newtheorem{lemma}[theorem]{Lemma}
\newtheorem{proposition}[theorem]{Proposition}
\newtheorem{corollary}[theorem]{Corollary}
\theoremstyle{definition}
\newtheorem{definition}[theorem]{Definition}
\newtheorem{remark}[theorem]{Remark}
\newcommand{\Hmm}[1]{\leavevmode{\marginpar{\tiny%
$\hbox to 0mm{\hspace*{-0.5mm}$\leftarrow$\hss}%
\vcenter{\vrule depth 0.1mm height 0.1mm width \the\marginparwidth}%
\hbox to 0mm{\hss$\rightarrow$\hspace*{-0.5mm}}$\\\relax\raggedright
#1}}}
\begin{document}

\title[The GPEF and condensation in the single particle ground
state]{The Gross-Pitaevskii functional with a random background
  potential and condensation in the single particle ground state}
\author[F. Klopp, B. Metzger]{Fr{\'e}d{\'e}ric Klopp, Bernd Metzger}
\address{Fr{\'e}d{\'e}ric Klopp, CNRS UMR LAGA 7539,
  Institut Galil{\'e}e, Universit{\'e} Paris 13, Paris, France\\ and \\
  Institut Universitaire de France}
\email{\href{mailto:klopp@math.univ-paris13.fr}{klopp@math.univ-paris13.fr}}
\address{Bernd Metzger, CNRS UMR LAGA 7539, Institut Galil{\'e}e,
  Universit{\'e} Paris 13, Paris, France}
\email{\href{mailto:metzger@math.univ-paris13.fr}{metzger@math.univ-paris13.fr}}
% \curraddr{} \email{Bernd.Metzger@ruhr-uni-bochum.de}
\thanks{The authors were partially supported by the grant
  ANR-08-BLAN-0261-01}

\begin{abstract}
  For discrete and continuous Gross-Pitaevskii energy functionals with
  a random background potential, we study the Gross-Pitaevskii ground
  state. We characterize a regime of interaction coupling when the
  Gross-Pitaevskii ground state and the ground state of the random
  background Hamiltonian asymptotically coincide.
\end{abstract}

\maketitle
 
% \tableofcontents
\section{Introduction}
\noindent The purpose of the present paper is to study some aspects of
condensation in the ground state of the Gross-Pitaevskii energy functional with a
disordered background potential.  As they can be treated very  similar, we
consider the  discrete and the  continuous setting simultaneously. \\[-2mm]
\begin{center}
  \begin{minipage}{12.2cm}
    {\bf The continuous setting:} \\[2mm]
    In $\R^d$, consider the cube $\Lambda_L=[-L,L]^d$ of side length
    $2L$ and volume $|\Lambda_L|=(2L)^d$. In
    $\mathcal{H}_L:=L^2(\Lambda_L)$, on the domain
    $\mathcal{D}_L:=H^2(\Lambda_L)$, consider ${\rm
      H}^P_{\omega,L}=(-\Delta+V_{\omega})^P_{\Lambda_L}$ the
    continuous self-adjoint Anderson model on $\Lambda_L$ with
    periodic boundary conditions. We assume
    \begin{itemize}
    \item $\Delta= \sum_{j=1}^d\partial^2_j$ is the continuous Laplace
      operator;
    \item $V_{\omega}$ is an ergodic random potential
      i.e. an ergodic random field over $\R^d$ that satisfies\\[-2mm]
      \begin{align*}
        \forall\alpha\in\N^d,\ \|\|\partial^\alpha
        V_\omega\|_{x,\infty}\|_{\omega,\infty}<+\infty\\[-2mm]
      \end{align*}
      where $\|\cdot\|_{x,\infty}$ (resp.
      $\|\cdot\|_{\omega,\infty}$) denotes the supremum norm in $x$
      (resp. $\omega$).
    \end{itemize}
    These assumptions are for example satisfied by a continuous
    Anderson model with a smooth compactly supported single site
    potential i.e. if
    \begin{equation*}
      V_\omega(x)=\sum_{\gamma\in\Z^d}\omega_\gamma u(x-\gamma)  
    \end{equation*}
    where $u\in\mathcal{C}_0^{\infty}(\R^d)$ and
    $(\omega_\gamma)_{\gamma\in\Lambda_L}$ are bounded, non negative
    identically distributed random variables.
  \end{minipage}
\end{center}
\begin{center}
  \begin{minipage}{12.2cm} {\bf The discrete setting:}\\[2mm] On the
    finite discrete cube $\Lambda_L=[-L,L]^d\cap\Z^d\subset\Z^d$ the
    cube of side length $2L+1$ and volume $|\Lambda_L|=(2L+1)^d$, let
    ${\rm H}^P_{\omega,L}=(-\Delta+V_{\omega})^P_{\Lambda_L}$ the
    discrete Anderson model on
    $\mathcal{D}_L=\mathcal{H}_L:=\ell^2(\Lambda_L)$ with periodic
    boundary conditions. We assume
    \begin{itemize}
    \item $\left( -\Delta\right)^P_{\Lambda_L}$ is the discrete
      Laplacian;
    \item $V_{\omega}$ is a potential i.e. a diagonal matrix entries
      of which are are given by bounded non negative random variables,
      say $\omega=(\omega_\gamma)_{\gamma\in\Lambda_L}$.
    \end{itemize}
  \end{minipage}
\end{center}
\noindent For the sake of definiteness, we assume that the infimum of
the (almost sure) spectrum of $H_\omega$ be $0$. We define
\begin{definition}[\textbf{Gross-Pitaevskii energy functional [GPEF]}]
  \label{discrete Gross-Pitaevskii model}
  The (one-particle) Gross-Pitaevskii energy functional on the cube
  $\Lambda_L$ (in the discrete or in the continuous) is defined by
  \begin{equation}
    \label{Gross-Pitaevskii}
    {\mathcal E}^{GP}_{\omega,L}[\varphi]= \langle {\rm H^P_{\omega,L}}
    \varphi,\varphi \rangle +{\rm U }\|\varphi\|_4^4
  \end{equation}
  for $\varphi\in\mathcal{D}_L$ and ${\rm U}$ is a positive coupling
  constant.
\end{definition}
\noindent For applications, it is natural that this coupling constant
is related to $|\Lambda_L|$. We refer to the discussion following
Theorem~\ref{Theorem1} for details. One proves
\begin{proposition}
  \label{pro:1}
  For any $\omega\in\Omega$ and $L\geq1$, there exists a ground state
  $\varphi^{\rm GP}$ i.e. a vector $\varphi^{\rm GP} \in
  \mathcal{D}_L$ such that $\|\varphi^{\rm GP}\|_2=1$   minimizing
  the  Gross-Pitaevskii energy functional, i.e.
  \begin{equation}
    \label{ground state}
    {\rm E}_{\omega,L}^{\rm GP}={\mathcal E}^{GP}_{\omega,L}[\varphi^{\rm
      GP}] =\mintwo{\varphi \in \mathcal{D}_L}{\|\varphi\|_2=1}
    {\mathcal E}^{GP}_{\omega,L}[\varphi].
  \end{equation}
  The ground state $\varphi^{\rm GP}$ can be chosen positive; it is
  unique up to a change of phase.
  ${\rm E}_{\omega,L}^{\rm GP}$ denotes the ground state energy of the
  discrete Gross-Pitaevskii functional.
\end{proposition}
\noindent The proof in the continuous case is given
in~\cite{Lieb-Buch}; the proof in the discrete case is similar.\\[2mm]
Let $H^N_{\omega,L}$ and $H^D_{\omega,L}$ respectively denote the
Neumann and Dirichlet restrictions of $H_\omega$ to $\Lambda_L$.  Our
main assumptions on the random model are: 
\begin{center}
  \begin{minipage}{12.2cm} {\bf (H0) Decorrelation estimate:} the model satisfies a finite
    range decorrelation estimate i.e. there exits $R>0$ such that, for
    any $J\in\N^*$ and any sets $(D_j)_{1\leq j\leq J}$, if
    \begin{align*}
    \D\inf_{j\not=j'}{\text dist}(D_j,D_{j'})\geq R,
    \end{align*}
    then the restrictions
    of $V_\omega$ to the domains $D_j$, i.e. the functions
    $(V_{\omega|D_j})_{1\leq j\leq J}$, are independent random fields. \\
    {\bf (H1) Wegner estimate:} There exists $C>0$
    such that, for any compact interval $I$ and $\bullet\in\{P,N,D\}$,
    \begin{equation*}
      \esp[\tr(\car_{I}(H^\bullet_{\omega,L}))]\leq C|I|L^d;
    \end{equation*} $\;$\\[-2mm]
    {\bf (H2)    Minami estimate: } There exists $C>0$
    such that, for $I$ a compact interval and $\bullet\in\{P,N,D\}$,
    \begin{equation*}
      \pro[\{H^\bullet_{\omega,L}\text{ has at least two eigenvalues in
      }I\}]\leq C(|I|L^d)^2; 
    \end{equation*} $\;$\\[-2mm]
    {\bf (H3) Lifshitz type estimate near energy ${\mathbf 0}$:}
    There exist constants $C>c>0$ such that, for $L\geq1$ and any
    parallelepiped $P_L=I_1\times\cdots\times I_d$ where the intervals
    $(I_j)_{1\leq j\leq d}$ satisfy $L/2\leq |I_j|\leq 2 L$, one has
    \begin{gather*}
      c e^{- L^d/c}\leq \pro[\{H^D_{\omega|P_L}\text{ has at least one
        eigenvalue in }[0,L^{-2}]\}],\\
      \pro[\{H^N_{\omega|P_L}\text{ has at least one eigenvalue in
      }[0,L^{-2}]\}]\leq C e^{-L^d/C}
    \end{gather*}
    where $H^D_{\omega|P_L}$ (resp. $H^D_{\omega|P_L}$) is the
    Dirichlet (resp. Neumann restriction) of $H_\omega$ to $P_L$.
    $\;$\\[-4mm]
  \end{minipage}
\end{center}
\noindent Let us now discuss the validity of these assumptions.\\
The decorrelation assumption (H0) is satisfied for the discrete
Anderson model described above if the random variables
$(\omega_\gamma)_{\gamma\in\Z^d}$ are i.i.d. (H0) clearly allows  
some correlation between the random variables. For the continuous
Anderson model, it is satisfied if the single site potential has
compact support and the random variables are i.i.d.\\
Under the assumption that the random variables are i.i.d and that their
distribution is regular, it is well known that the Wegner estimate
(H1) holds at all energies for both the discrete and continuous
Anderson model (see e.g.~\cite{KiMe,MR2378428,MR2362242}).\\
The Minami estimate (H2) is known to hold at all energies under
similar regularity assumptions for the discrete Anderson model (see
e.g.~\cite{MR97d:82046,MR2360226,MR2290333,MR2505733}) and for the continuous Anderson model in the
localization regime under more specific assumptions on the single site
potential  (see e.g.~\cite{CGK2}).\\
Finally, the Lifshitz tails estimate (H3) is known to hold for both
the continuous and discrete Anderson model under the sole assumption
that the i.i.d. random variables be non degenerate, non negative and
$0$ is in their essential range (see
e.g.~\cite{Ki-Ma:83b,MR833221,Kirsch}). Though the Lifshitz tails
estimate is usually not stated for parallelepipeds but for cubes, the
proof for cubes applies directly to parallelepipeds satisfying the
condition stated in (H3).
\vskip.4cm\noindent The main result of the present paper is
\begin{theorem}
  [\textbf{Condensation in the single particle ground state}]
  \label{Theorem1}
  Assume assumptions (H0)-(H3) hold. Denote by $\varphi_0$ the single
  particle ground state of ${\rm H}^P_{\omega,L}$ (chosen to be
  positive for the sake of definiteness) and by $\varphi^{\rm GP}$ the
  Gross-Pitaevskii ground state. \\
  If for $L$ large, one assumes that
  \begin{equation*}
    U=U(L)=   o\left(\frac{L^{-d}}{(1+(\log L)^{d-2/d+\epsilon})f_d(\log
        L)}\right)
  \end{equation*}
  where
  \begin{equation}
    \label{eq:12}
    f_d(\xi)=
    \begin{cases}
      \xi^{-1/4}&\text{ if }d\leq3,\\ \xi^{-1/d}\log\xi&\text{ if }d=4,\\
      \xi^{-1/d}&\text{ if }d\geq5.
    \end{cases}
  \end{equation}
  and $\epsilon=0$ in the discrete setting, resp.  $\epsilon>0$
  arbitrary in the continuous case,
  then, there exists $0<\eta(L)\to0$ when $L\to+\infty$ such that
  \begin{equation}
    \label{eq:13}
    \bP[ |\langle\varphi_0,\varphi^{\rm GP}\rangle-1|\geq
    \eta(L)\}]\vers_{L\to+\infty}0.
  \end{equation}$\;$\\[-4mm]
\end{theorem}
\noindent The proof of Theorem~\ref{Theorem1} also yields information
on the size of $\eta(L)$ and on the probability estimated
in~\eqref{eq:13}. Note that the assumption (H1)-(H3) can be relaxed at
the expense of changing the admissible size for $U$.\\[2mm]
To appreciate Theorem~\ref{Theorem1} maybe some comments about the
physical background of the Gross-Pitaevskii model, its relationship to
Bose-Einstein condensation and to known results are of interest.
Motivated by recent experiments with weakly interacting Bose gases in
optical lattices (see for example~\cite{BlDaZW}) the fundamental
objects of interest are the ground state density and energy, i.e.
\begin{equation}
  \label{ground state energy} 
  \mathcal{E}^{\rm QM}:=\min_{\substack{\Phi\in
      \underset{s}{\overset{N}{\otimes}} L^2(\Lambda_L)\\\|\Phi\|=1}}
  \left\langle\Phi,\left[\sum_{i=1}^{N}\{-\Delta_i +V(x_i)\} + \sum_{1
        \leq i<j\leq N} v(|x_i-x_j|)\right]\Phi\right\rangle.
\end{equation}
The optical lattice is modeled by the background potential $V$ as
shown in Figure~\ref{fig:1}.
\begin{figure}%[hp]
  \begin{center}
    \includegraphics[width=8cm]{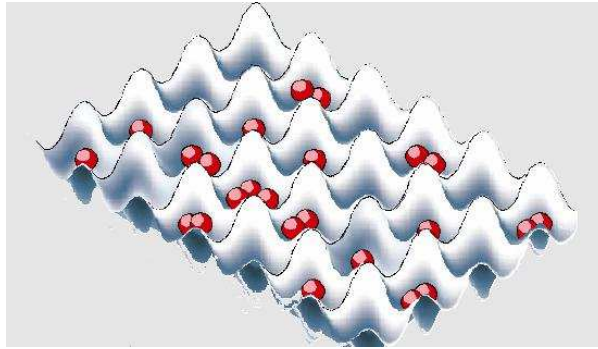}
  \end{center}
  \label{fig:1}
  \caption{An example of a background potential modeling an optical
    lattice \cite{San} }
\end{figure}
\noindent Assuming a weak interaction limit of the interaction
potential $v(x,y)$, the continuous $N$-particle Gross-Pitaevskii
energy functional
\begin{align}
  \label{GP ground state energy}
  \mathcal{E}^{\rm GP}= \min_{\substack{\varphi\in
      L^2(\Lambda_L)\\\|\varphi\|_2=1}} \int_{\Lambda_L}\left(
    N|\nabla\varphi(x)|^2+ N V|\varphi(x)|^2+4 N^2 \pi
    \mu a|\varphi(x)|^4\right)dx,
\end{align}
is a mean field approximation of the ground state energy (\ref{ground
  state energy}), e.g. in three dimensions one has
\begin{align*}
  \lim_{N\to\infty}\frac{{\mathcal{E}^{\rm QM}}}{ {\mathcal{E}^{\rm
        GP}}}=1.
\end{align*}
( see for example~\cite{Lieb-Buch} ) The discrete Gross-Pitaevskii
model is then a tight binding approximation of the continuous
one-particle Gross-Pitaevskii functional~\cite{SmTr1,SmTr2}
\begin{equation*}
  \mathcal{E}^{\rm GP}[\varphi]=\int_{\bR^3}\left(
    |\nabla\varphi(x)|^2+ V|\varphi(x)|^2+4 N \pi \mu
    a|\varphi(x)|^4\right)dx.
\end{equation*}
Another way to derive the discrete Gross-Pitaevskii model starts with
a discretization of (\ref{ground state energy}) yielding the standard
description of optical lattices using the Bose-Hubbard-Hamiltonian
\begin{align*}
  H=-\sum_{|n-n'|=1} c_n^{\dagger}c_{n'} +\sum_n(\sigma
  V_n-\mu)n_n+\frac{1}{2}U\sum n_n^2
\end{align*}
where $c_n^{\dagger}$, $c_{n}$ are bosonic creation and annihilation
operators and $n_n$ gives the particle number at site $n$ (see the
survey article~\cite{BlDaZW} and references therein). A mean field
approximation then yields the discrete Gross-Pitaevskii energy
functional~\cite{Lee}.\\[4mm]
One motivation to study Bose gases is Bose-Einstein condensation, i.e.
the phenomena that a single particle level has a macroscopic
occupation ( a non-zero density in the thermodynamic limit)
\cite{Lieb-Buch}.  Introduced in~\cite{Einstein} in the context of an
ideal Bose gas, it was due to naturally arising interactions a
difficult problem to realize Bose-Einstein condensation experimentally
\cite{Daetal,Ketterle}. \\
As we will see, also the formal description is more
elaborated. To motivate the definition of BEC for vanishing
temperature we
follow the continuous approach in~\cite{Lieb-Buch}.  To formalize the
concept of a macroscopic occupation of a single particle state we
remember the definition of the one-particle density matrix
\cite{Lieb-Buch}, i.e.  the operator on $L^2(\bR^3)$ given by the
kernel
\begin{align*}
  \gamma(x,x')=N\int \Phi^{\rm QM}(x,x_2,\dots,x_N) \Phi^{\rm
    QM}(x',x_2,\dots,x_N) \prod\limits_{j= 2}^N dx_j
\end{align*}
with the normalized ground state wave function $\Phi^{\rm QM}$ of the
many Boson Hamiltonian.  BEC in the ground state is then defined that
the projection operator $\gamma$ has an
eigenvalue of order $N$ in the thermodynamic limit. \\[2mm]
Remembering that for the ideal Bose gas the multi-particle ground
state can be represented as a product
\begin{align*}
  \Phi^{\rm
    QM}(x_{1},\dots,x_{N})=\hbox{$\prod_{i=1}^{N}$}\;\varphi_{0}(x_{i})
\end{align*}
of the single particle ground state $\varphi_{0}$ the one-particle
density matrix becomes
\begin{align*}
  \gamma(x,x')= \; N \:\varphi_{0}(x)\varphi_{0}(x') ,
\end{align*}
thus the definition of BEC above is natural and can also be related to
the thermodynamic formalism (see e.g.~\cite{LePaZa,Lieb-Buch} and
references).  In particular, it is of interest to consider BEC for the
ideal Bose gas with a random background potential. In this case the
Lifshitz tail behavior at the bottom of the spectrum makes a
generalized form of Bose-Einstein condensation possible even for $d =
1$, $2$ (see~\cite{LePaZa} and references cited there).\\[2mm]
The situation in the Gross-Pitaevskii-limit is close to the situation
for the ideal Bose gas~\cite{Lieb-Buch}.  The one-particle density
matrix is asymptotically given by
\begin{align} \label{Condensation in the ground state} \gamma(x, x')
  \stackrel{N\rightarrow \infty}{\sim } \; N \:\varphi^{\rm
    GP}(x)\varphi^{\rm GP}(x') .
\end{align}
Physically the content of (\ref{Condensation in the ground state}) is
that all Bose particles will condensate in the GP ground state
motivating the definition of complete (or 100\%)
BEC in~\cite{Lieb-Buch}.\\[4mm]
The purpose of the present publication is a first step to analyze the
fine structure of the Gross-Pitaevskii ground state. Under the
assumption of a random background potential we want to understand how
$\varphi^{\rm GP}$ is related to the eigenstates of the single
particle Hamiltonian. More familiar is this problem in the following
two   settings. \\[2mm]
If the Bosons are trapped by a potential tending to $\infty$,
i.e. $\liminf_{|x|\rightarrow \infty}V(x) =\infty$, the spectral
properties of the single particle are invariant in the thermodynamic
limit, i.e. the discrete spectrum and the strictly positive distance
between the first two eigenvalues. Assuming $Na \rightarrow 0$ in the
continuous setting, respectively $NU \rightarrow 0$ in the context of
the discrete Gross-Pitaevskii model, the interaction energy is a small
perturbation of the single particle energy functional. In this
situation it is natural, that in the thermodynamic limit $\varphi^{\rm
  GP}$ and the single particle ground state $\varphi_{0}$
coincide~\cite{Lewin}.\\[2mm]
A complementary situation is given if the Bosons are confined to a
cube $\Lambda_L$ with $|\Lambda_L|\rightarrow \infty$ but without a
background potential. As described in~\cite{Lieb-Buch} assuming
$\rho=N/L^3$ and $ g=Na/L$ in the limit $N\to\infty$ one can prove
\begin{align*}
  \lim_{N\to\infty} \frac{1}{N} \frac{1}{L^3} \int\!\!\!\int
  \gamma(x,\, y) dx dy = 1 ,
\end{align*}
i.e. BEC in the normalized single particle ground state
$\varphi_{0}=L^{-d/2}\chi_{\Lambda_L}$.  As explained in
\cite{Lieb-Buch}
\begin{align*}
  g=\frac{Na}{L}=\frac{\rho a}{1/L^2}
\end{align*}
is in this context the natural interaction parameter since `` in the
GP limit the interaction energy per particle is of the same order of
magnitude as the energy gap in the
box, so that the interaction is still clearly visible''.  \\[2mm]
As emphasized in the physics literature (see
e.g.~\cite{BlDaZW,LuClBoAsLeSP}), new phenomena like fragmented BEC
(Lifshitz glasses) should occur when Bosons are trapped in a random
background potential. Our purpose in this publication is more modest.
We want to understand the natural interaction parameter in a random
media, s.t. the Gross-Pitaevskii ground state is close to the ground
state of the single particle Hamiltonian as it is suggested by the
situation in the ideal Bose gas.  As we will see the setting of Bosons
trapped in a random potential is not really comparable to the two
situations described above.\\
Under our assumptions, near $0$ which is almost sure limit of
$\inf(H^P_{\omega,L})$, we are in the localized regime, i.e. one has
pure point spectrum and localized eigenfunctions. In contrast to the
situation with vanishing potential the eigenstates close to the bottom
of the spectrum are localized in a small part of $\Lambda_L$, i.e. the
interaction energy will be larger than in the case of the homogeneous
Bose gas. In the random case, we determine the almost sure behavior of
the ground state from information on the integrated density of states
(see Lemma~\ref{Lemma1}). Under our weak Lifshitz tails assumption
(H3), we obtain that the ground state energy is of size $(\log
L)^{-2/d}$.  When $L\to+\infty$, the difference between the first two
eigenvalues will tend to zero; the speed at which this happens is
crucial in our analysis (see Proposition~\ref{le:2}).  In our case, we
estimate that, with good probability, it must be at least of order
$L^{-d}$. This difference is much smaller than the one obtained in the
homogeneous Bose gas where it typically is of order $L^{-2}$. We deem
that the estimate $L^{-d}$ for the spacing is not optimal in the
present setting. This estimate is the correct one in the bulk of the
spectrum; at the edges, the spacings should be larger. It seems that
getting an optimal estimate requires a much better knowledge of the
integrated density of states or, in other words, much sharper Lifshitz
tails type estimates (see (H3)) and Minami type estimates that take
into account the fact that we work at the edge of the spectrum (see
(H2)). Combining these observations explains the interaction parameter
$U=o(L^{-d}h_d^{-1}(\log L))$ that we don't believe to be optimal.\\[2mm]
Let us now briefly outline the structure of our paper. To prove our
result we need two ingredients. We need an upper bound of the
interaction term, i.e. we have to estimate the $\|.\|_4$- norm of the
single particle ground state $\varphi_0$. At the same time, we need a
lower bound of the distance of the first two single particle
eigenvalues asymptotically almost surely (a.a.s.) i.e. with a
probability tending to $1$ in the thermodynamic limit. Comparing these
two estimates we will see that under the assumptions of Theorem
\ref{Theorem1} it is energetically favorable, that the
Gross-Pitaevskii ground state and the single particle ground state
coincide. This will be proven at the end of this
publication.\\
To estimate the interaction term we will prove in Lemma~\ref{Lemma1}
that almost surely in the thermodynamic limit the single particle
ground state is flat, i.e.
\begin{align*}\| \nabla \varphi_0\|^2 \xrightarrow{L \rightarrow
    \infty} 0 \qquad \text{a.a.s.}
\end{align*}
This then yields an estimate of the interaction term which is the
purpose of Proposition~\ref{Lemma2}.\\
The a.a.s. lower bound of the distance of the first two single
particle eigenvalues is a little bit more intricate and uses the
Wegner and Minami estimates; it is related to the methods developed
in~\cite{GeKl}.  In Lemma~\ref{Lemma3}, we first estimate the
probability that the first two eigenstates and also their localization
center are close together. If the localization centers are relatively
far away, one can decouple the eigenstates and treat the first two
eigenvalues of each other. This is used in Lemma~\ref{Lemma5}.
\section{Estimating the interaction term}
\label{sec:estim-inter-term}
\noindent The main result of this section is an upper bound on
${\rm E}_{\omega,L}^{\rm GP}- {\rm E_0^P[\omega,L]}$. This quantity is
non negative (see~(\ref{ground state})) and we prove
\begin{proposition}
  \label{Lemma2}
  There exists $C>0$, such that, for any $p\in\N$, one has
  \begin{equation}
    \label{eq:4}
    \bP\left[{\rm E}_{\omega,L}^{\rm GP} - {\rm E_0^P[\omega,L]} \leq  C U
      f_d(\log L)\right] \geq 1 -L^{-p}
  \end{equation}
  where $f_d$ is defined in~\eqref{eq:12}.\\[-2mm]
\end{proposition}
\noindent By definition, for $\varphi_0(\omega,L)$ the ground state of
$H^P_{\omega,L}$, one has
\begin{equation}
  \label{eq:5}
  {\rm E}_{\omega,L}^{\rm GP}\leq {\mathcal
    E}_{\omega,L}[\varphi_0(\omega,L)]={\rm E_0^P[\omega,L]}+ U
  \|\varphi_0(\omega,L)\|_4^4
\end{equation}
resp.
\begin{equation*}
  {\rm E}_{\omega,L}^{\rm GP} - {\rm E_0^P[\omega,L]} \leq   U
  \|\varphi_0(\omega,L)\|_4^4.
\end{equation*}
To prove Proposition \ref{Lemma2}, resp. control the interaction term,
we first estimate the ground state energy of the random Schr{\"o}dinger
operator and derive in Corollary \ref{le:4} an estimate on the
``flatness'' of its ground state. We start with the Dirichlet and
Neumann boundary cases.
\begin{lemma}
  \label{Lemma1}
  Assume (H3) is satisfied. Let $E_0^P(\omega,L)$ be the ground state
  energy of $H^P_{\omega,L}$ and denote by $\varphi_0(\omega,L)$ the
  associated positive normalized ground
  state.\\
  Then, for any $p>0$, there is a constant $C>0$ such that, for $L$
  sufficiently large,
  \begin{equation}
    \label{eq:1}
    \bP\left[ C^{-1} (\log L)^{-2/d}\leq E_0^P(\omega,L) \leq C (\log
      L)^{-2/d} \right]\geq 1 -L^{-p}.
  \end{equation}$\;$\\[-4mm]
\end{lemma}
\noindent As $V_\omega$ is non negative and $\varphi_0(\omega,L)$
normalized, one has $\| \nabla\varphi_0(\omega,L)\|^2\leq
E_0^P(\omega,L)$. Hence, Proposition~\ref{Lemma2} implies the following
``flatness'' estimate of the ground state.
\begin{corollary}
  \label{le:4}
  Under the assumptions of Proposition~\ref{Lemma2}, for any $p>0$,
  there is a constant $C>0$ such that, for $L$ sufficiently large,
  \begin{equation}
    \label{eq:2}
    \bP\left[ \| \nabla\varphi_0(\omega,L)\|^2 \leq C (\log L)^{-2/d}
    \right]\geq 1 -L^{-p}.
  \end{equation}$\;$\\[-2mm]
\end{corollary}
\noindent It is maybe interesting to note that from a Lifshitz tail type
estimate (i.e. the annealed estimate), we recover the (approximate)
almost sure behavior of the ground state energy of $H^N_{\omega,L}$
(i.e. the quenched estimate) (see e.g.~\cite{MR2001h:60147}).\\
\noindent We  note that Proposition~\ref{Lemma2} and Corollary
\ref{le:4} also hold if we replace the periodic ground state and
ground state energy by the Neumann or Dirichlet ones.
\begin{proof}[Proof of Lemma~\ref{Lemma1}]
  Fix $\ell\geq1$. Decompose the interval $[-L,L]$ into intervals of
  length comprised between $\ell/2$ and $2\ell$. This yields a
  partition of $\Lambda_L$ in parallelepipeds i.e.
  \begin{equation*}
    \Lambda_L=\bigcup_{1\leq j\leq J}P_j
  \end{equation*}
  such that
  \begin{itemize}
  \item $P_j=I_j^1\times\cdots\times I_j^d$ where the intervals
    $(I^k_j)_{1\leq k\leq d}$ satisfy $\ell/2\leq |I^j_k|\leq 2\ell$
  \item for $j\not=j'$, $P_j\cap P_{j'}=\emptyset$,
  \item $J$, the number of parallelepiped, satisfies
    $2^{-d}(L/\ell)^d\leq J\leq 2^d(L/\ell)^d$.
  \end{itemize}
  In the continuous model, one can take the parallel piped to be cubes.\\
  Denote by $\omega_{|\Lambda_L}$ the restriction of $\omega$ to
  $\Lambda_L$. Furthermore, let $\omega^{P,L}$ be the periodic
  extension of $\omega_{|\Lambda_L}$ to $\Z^d$ i.e. for
  $\beta\in\Lambda_L$ and $\gamma\in\Z^d$, $\omega_{\beta+\gamma
    \overline{L}}^{P,L}=\omega_\beta$ where $\overline{L}=2L+1$ in the
  discrete case and $2L$ in the continuous one. As ${\rm
    H}^P_{\omega}$ is the periodic restriction of $H_\omega$ to
  $\Lambda_L$, we know that ${\rm
    E_0^P[\omega,L]}=\inf\sigma(H_{\omega^{P,L}})$ where this last
  operator is considered as acting on the full space $\R^d$ or
  $\Z^d$ (see e.g.~\cite{MR89b:35127}).\\
  We can now decompose $\R^d$ or $\Z^d$ into
  $\cup_{\gamma\in\Z^d}\cup_{j=1}^J(\gamma \overline{L}+P_j)$. By
  Dirichlet-Neumann bracketing (see e.g.~\cite{Kirsch,KiMe}) $H_{\omega^{P,L}}$
  satisfies as an operator on $\Z^d$ or $\R^d$ 
  \begin{equation}
    \label{eq:14}
    \oplus_{\gamma\in\Z^d}\oplus_{j=1}^J H^N_{\omega|(\gamma
      \overline{L}+P_j)} \leq H_{\omega^{P,L}}\leq
    \oplus_{\gamma\in\Z^d}\oplus_{j=1}^J 
    H^D_{\omega|(\gamma \overline{L}+P_j)}.
  \end{equation}
  Define
  \begin{equation}
    \label{eq:19}
    {\rm E_0^\bullet[\omega,\ell,j]}=\inf \sigma({\rm
      H}^\bullet_{\omega|P_j})\quad\text{for}
    \quad\bullet\in\{N,D\}; 
  \end{equation}
  here, the superscripts $D$ and $N$ refer respectively to the
  Dirichlet and Neumann boundary conditions. As $\omega^{P,L}$ is
  $\overline{L}\Z^d$-periodic, ${\rm H}^\bullet_{\omega|P_j}$ and
  ${\rm H}^\bullet_{\omega|(\gamma\overline{L}+P_j)}$ are unitarily
  equivalent. The bracketing~(\ref{eq:14}) then yields
  \begin{equation*}
    \inf_{1\leq j\leq J} {\rm E_0^N[\omega,\ell,j]}\leq {\rm
      E_0^P[\omega,L]} \leq \inf_{1\leq j\leq J} {\rm E_0^D[\omega,\ell,j]}.
  \end{equation*}
  Labeling every second interval of the partition of $[-L,L]$ used to
  construct the partition of $\Lambda_L$, we can partition the
  interval $\{1,\cdots,J\}$ into $2^d$ sets, say
  $(\mathcal{J}_l)_{1\leq l\leq 2^d}$ such that
  \begin{enumerate}
  \item if $l\not= l'$, $\mathcal{J}_l\cap\mathcal{J}_{l'}=\emptyset$,
  \item for $j\in \mathcal{J}_l$ and $j'\in \mathcal{J}_l$ such that
    $j\not= j'$, one has dist$(P_j,P_{j'})\geq\ell/2$,
  \item there exists $C>0$ such that for $1\leq l\leq 2^d$,
    $C^{-1}(L/\ell)^d\leq \#\mathcal{J}_l\leq C(L/\ell)^d$.
  \end{enumerate}
  Assume $R$ is given by (H0).
  By (2) of the definition of the partition above, for any $l\geq 2R$, all the $({\rm
    H}^\bullet_{\omega|P_j})_{j\in \mathcal{J}_l}$, resp. all the
  $(E_0^\bullet[\omega,\ell,j])_{j\in \mathcal{J}_l}$ (for
  $\bullet\in\{N,D\}$) are independent. Hence, using~(\ref{eq:19}), we
  compute
  \begin{equation*}
    \begin{split}
      \bP[ {\rm E_0^P[\omega,L]}> E]&\leq \bP[ \inf_j {\rm
        E_0^D[\omega,\ell,j]}> E]\leq
      \sum_{l=1}^{2^d}\prod_{j\in\mathcal{J}_l}\bP[ {\rm
        E_0^D[\omega,\ell,j]}> E]\\
      &= \sum_{l=1}^{2^d}\prod_{j\in\mathcal{J}_l}\left( 1-\bP[ {\rm
          E_0^D[\omega,\ell,j]}\leq E]\right).
    \end{split}
  \end{equation*}
  Pick $E=c\ell^{-2}$ where $c$ is given by assumption (H3) and
  \begin{equation*}
    (k\log L-c^{-1}\log c)^{1/d}\leq\ell\leq
    (k\log L-c^{-1}\log c)^{1/d}+1
  \end{equation*}
  where $k$ will be chosen below. Applying the Lifshitz estimate (H3),
  we obtain
  \begin{equation*}
    \begin{split}
      \bP[ {\rm E_0^P[\omega,L]}> E]&\leq \sum_{l=1}^{2^d}\left(
        1- e^{-k \log L/c}\right)^{\#\mathcal{J}_l} \\&\leq
      \sum_{l=1}^{2^d}\exp\left(-\#\mathcal{J}_l\,e^{-k \log
          L/c}\right)\leq O(L^{-\infty})
    \end{split}
  \end{equation*}
  if we choose $k<c d$ as $C^{-1}(L/\ell)^d\leq \#\mathcal{J}_l\leq
  C(L/\ell)^d$ for $1\leq l\leq 2^d$.\\
  Hence, we have
  \begin{equation*}
      \bP[ {\rm E_0^N[\omega,L]}\leq E]\geq 1-O(L^{-\infty}).
  \end{equation*}
  To estimate from below, we use again~(\ref{eq:19}) to get
  \begin{equation*}
    \bP[ {\rm E_0^P[\omega,L]}\leq E] \leq 
    \sum_{j\in\mathcal{J}}
    \bP[ {\rm E_0^N[\omega,\ell,j]}\leq E]).
  \end{equation*}
  Pick $E=C\ell^{-2}$ where $C$ is given by assumption (H3) and
  \begin{equation*}
    (k\log L-C^{-1}\log C)^{1/d}\leq\ell\leq
    (k\log L-C^{-1}\log C)^{1/d}+1
  \end{equation*}
  where $k$ will be chosen below. As $\#\mathcal{J}\leq C(L/\ell)^d$,
  applying the Lifshitz estimate (H3), we obtain
  \begin{equation*}
    \bP[ {\rm E_0^N[\omega,L]}\leq E]\leq
    C\left(\frac{L}{\ell}\right)^d e^{-k \log
      L/C}\leq L^{-p}
  \end{equation*}
  if we choose $k>(d+p)C$. Hence, we have
  \begin{equation*}
    \bP[ {\rm E_0^N[\omega,L]}\geq E]\geq 1-L^{-p}.
  \end{equation*}
  This completes the proof of Lemma~\ref{Lemma1}.
\end{proof}
\noindent To prove estimate~(\ref{eq:4}), we will use the spectral
decomposition of $-\Delta_L^P$. Though the arguments in the discrete
and continuous cases are quite similar, it simplifies the discussion
to distinguish between the discrete and the continuous case rather
than to introduce uniform notations. We start with the discrete case.
\begin{lemma}
  \label{le:1}
  There exists $C>0$ such that, for $\varepsilon\in(0,1)$ and $L\in\N$
  satisfying $L\cdot\varepsilon\geq1$ one has for
  $u\in\ell^2(\Z^d/(2L+1)\Z^d)$ with $\|u\|_2=1$ and $\langle
  -\Delta_L^Pu,u\rangle\leq\varepsilon^2$ the estimate
  \begin{equation}
    \label{eq:6}
    \|u\|_4\leq C g_d(\varepsilon)\text{ where }
    g_d(\xi)=
    \begin{cases}
      \xi^{d/4}&\text{ if }d\leq 3,\\ \xi|\log\xi|&\text{ if }d=4,\\
      \xi&\text{ if }d\geq5.
    \end{cases}
  \end{equation}
\end{lemma}
\begin{proof}
  The spectral decomposition of $-\Delta_L^P$ is given by the discrete
  Fourier transform that we recall now. Identify $\Lambda_L$ with the
  Abelian group $\Z^d/(2L+1)\Z^d$. For $u\in\mathcal{H}_L$, set
  \begin{equation}
    \label{eq:9}
    \hat u=(\hat u_\gamma)_{|\gamma|\leq L}\text{ where }
    \hat u_\gamma=\frac{1}{(2L+1)^{d/2}}\sum_{|\beta|\leq
      L}u_\beta\cdot e^{-2i\pi\gamma\beta/(2L+1)}.
  \end{equation}
  Then, one checks that (see e.g.~\cite{Klopp1})
  \begin{equation}
    \label{eq:7}
    (-\Delta_L^Pu)\hat\ =(h(\gamma)\hat u_\gamma)_{|\gamma|\leq
      L}\text{ where }h(\gamma)=2d-2\sum_{j=1}^d
    \cos\left(\frac{2\pi\,\gamma_j}{2L+1}\right).
  \end{equation}
  Pick $u\in\ell^2(\Z^d/(2L+1)\Z^d)$ with $\|u\|_2=1$ and $\langle
  -\Delta_L^Pu,u\rangle\leq\varepsilon^2$ and write
  $u=\sum_{k=0}^{k_\varepsilon}u_k$ where $k_\varepsilon\in\N$,
  $-\log\varepsilon\leq k_\varepsilon<-\log\varepsilon+1$ and
  \begin{itemize}
  \item $\hat u_0=\hat u\cdot \car_{|\gamma|< \varepsilon L}$
  \item for $1\leq k\leq k_\varepsilon-1$, $\hat u_k=\hat u\cdot
    \car_{e^{k-1}\varepsilon L \leq |\gamma|< e^k\varepsilon L}$
  \item $\hat u_{k_\varepsilon}=\hat u\cdot
    \car_{e^{k_\varepsilon}\varepsilon L\leq |\gamma|}$
  \end{itemize}
  where $\hat u$ denotes the discrete Fourier transform defined
  defined in~\eqref{eq:9}.\\
  Then, for $k\not=k'$, $\langle u_k,u_{k'}\rangle=0$ and,
  using~(\ref{eq:7}), for $k\geq1$,
  \begin{equation*}
    C^{-1}\sum_{k=0}^{k_\varepsilon}(e^{k-1}\varepsilon)^2\|u_k\|_2^2
    \leq\sum_{k=0}^{k_\varepsilon}\langle-\Delta_L^P
    u_k,u_k\rangle=\langle-\Delta_L^P u,u\rangle\leq\varepsilon^2
  \end{equation*}
  i.e.
  \begin{equation}
    \label{eq:8}
    \sum_{k=0}^{k_\varepsilon}e^{2k}\|u_k\|^2_2\leq C.
  \end{equation}
  Hence, using~(\ref{eq:9}) and H{\"o}lder's inequality, we compute
  \begin{equation*}
    \begin{split}
      |(u_k)_\beta|&=\frac1{(2L+1)^{d/2}}\left|\sum_{e^{k-1}\varepsilon
          L \leq |\gamma|< e^k\varepsilon L} (\hat u_k)_\gamma
        e^{-2i\pi\gamma\beta/(2L+1)} \right|\\
      &\leq \frac1{(2L+1)^{d/2}}\left(\sum_{e^{k-1}\varepsilon L \leq
          |\gamma|< e^k\varepsilon L} |(\hat
        u_k)_\gamma|^p\right)^{1/p} \left(\sum_{e^{k-1}\varepsilon
          L \leq |\gamma|< e^k\varepsilon L}1\right)^{1/q}\\
      &\leq \frac1{(2L+1)^{d/2}}\|\hat
      u_k\|_p\,(e^k\varepsilon(2L+1))^{d/q}.
    \end{split}
  \end{equation*}
  So, for $p=q=2$, one gets
  \begin{equation}
    \label{eq:15}
    \|u_k\|_\infty\leq\|u_k\|_2\,(e^k\varepsilon)^{d/2}.
  \end{equation}
  Then, using~(\ref{eq:8}), we compute
  \begin{equation*}
    \|u\|_4\leq\sum_{k=0}^{k_\varepsilon}\|u_k\|_4
    \leq\sum_{k=0}^{k_\varepsilon}\sqrt{\|u_k\|_2\|u_k\|_\infty}
    \leq C\sum_{k=0}^{k_\varepsilon}e^{k(d-4)/4}
    \varepsilon^{d/4}\leq C g_d(\varepsilon)
  \end{equation*}
  where $g_d$ is defined in~(\ref{eq:6}). This completes the proof of
  Lemma~\ref{le:1}.
\end{proof}
\begin{remark}
  \label{rem:1}
  Lemma~\ref{le:1} is essentially optimal as, for $L$ sufficiently
  large,
  \begin{itemize}
  \item the trial function
    \begin{equation*}
      u_\gamma=
      \begin{cases}
        \varepsilon\text{ if }\gamma=0,\\
        (2L+1)^{-d/2}\text{ if }\gamma\not=0,
      \end{cases}
    \end{equation*}
    satisfies $1\leq \| u\|_2\leq 1+\varepsilon$, $\langle-\Delta_L^P
    u,u\rangle\leq C\varepsilon^2$ and $\|u\|_4\geq \varepsilon/C$;
  \item the trial function
    \begin{equation*}
      \hat u_\gamma=
      \begin{cases}
        (2\varepsilon L+1)^{-d/2}\text{ if }|\gamma|\leq\varepsilon
        L,\\ 0\text{ if }|\gamma|>\varepsilon L,
      \end{cases}
    \end{equation*}
    satisfies $\| u\|_2=1$, $\langle-\Delta_L^P u,u\rangle\leq
    C\varepsilon^2$ and $\|u\|_4\geq \varepsilon^{d/4}/C$.
  \end{itemize}
\end{remark}
\noindent We now turn to the continuous case.
\begin{lemma}
  \label{le:3}
  Fix $\eta\in(0,1/4)$. There exists $C>0$ such that, for
  $\varepsilon\in(0,1)$, $n>(d-2)\eta^{-1}+1$ and $L\in\N$ satisfying
  $L\cdot\varepsilon\geq1$ one has for $u\in H^n(\R^d/(2L)\Z^d)$ with
  $\langle -\Delta_L^Pu,u\rangle\leq\varepsilon^2$ the norm estimate
  \begin{equation*}
    \|u\|_4\leq Cg_{d,n}(\varepsilon)\|u\|^{\eta}_{H^n}\text{ where }
    g_{d,\eta}(\xi)=
    \begin{cases}
      \xi^{d/4}&\text{ if }d\leq 3,\\ \xi^{1-\eta}|\log\xi|&\text{ if }d=4,\\
      \xi^{1-\eta}&\text{ if }d\geq5.
    \end{cases}
  \end{equation*}
\end{lemma}
\begin{proof}
  We now use the Fourier series transform to decompose
  $-\Delta_L^P$. Identify $\Lambda_L$ with the Abelian group
  $\R^d/2L\Z^d$. For $u\in\mathcal{H}_L$, set
  \begin{equation}
    \label{eq:16}
    \hat u=(\hat u_\gamma)_{\gamma\in\Z^d}\text{ where }
    \hat u_\gamma=\frac{1}{(2L)^{d/2}}\int_{\Lambda_L}
    u(\theta)\cdot e^{-\pi i\gamma\theta/L}d\theta.
  \end{equation}
  Then,
  \begin{equation}
    \label{eq:17}
    u(\theta)=\frac{1}{(2L)^{d/2}}\sum_{\gamma\in\Z^d}\hat u_\gamma
    e^{\pi i\gamma\theta/L}
  \end{equation}
  and
  \begin{equation}
    \label{eq:18}
    (-\Delta_L^Pu)\hat\ =\left(\left|\frac{\pi\gamma}L\right|^2\hat 
      u_\gamma\right)_{|\gamma|\leq L}\text{ if }u\in\mathcal{D}_L.
  \end{equation}
  Pick $u$ as in Lemma~\ref{le:3} and decompose it as in the proof of
  Lemma~\ref{le:1} i.e. write $u=\sum_{k=0}^{k_\varepsilon}u_k$ where
  $k_\varepsilon\in\N$, $-\log\varepsilon\leq
  k_\varepsilon<-\log\varepsilon+1$ and
  \begin{itemize}
  \item $\hat u_0=\hat u\cdot \car_{|\gamma|< \varepsilon L}$
  \item for $1\leq k\leq k_\varepsilon-1$, $\hat u_k=\hat u\cdot
    \car_{e^{k-1}\varepsilon L \leq |\gamma|< e^k\varepsilon L}$
  \item $\hat u_{k_\varepsilon}=\hat u\cdot
    \car_{e^{k_\varepsilon}\varepsilon L\leq |\gamma|}$
  \end{itemize}
  where $\hat u$ denotes the Fourier series transform defined
  in~\eqref{eq:16} and~\eqref{eq:17}.\\
  The control on $u_k$ for $0\leq k\leq k_\varepsilon-1$ is obtained
  in the same way as in the proof of Lemma~\ref{le:1} namely the
  estimate~\eqref{eq:15} holds for $0\leq k\leq k_\varepsilon-1$. The
  additional ingredient that we need is to obtain a control over the
  large frequency components.\\
  Recall that
  \begin{equation*}
    \|u\|^2_{H^n}=\sum_{\gamma\in\Z^d}
    \left(1+\left|\frac{\pi\gamma}{L}\right|^2\right)^{n/2}|\hat u_\gamma|^2
  \end{equation*}
  Fix $r>d$. For notational convenience, write $v=u_{k_\varepsilon}$
  and compute
  \begin{equation*}
    \begin{split}
      |v(\theta)|&=
      \frac1{(2L)^{d/2}}\left|\sum_{e^{k_\varepsilon}\varepsilon L
          \leq |\gamma|} (\hat v)_\gamma
        e^{-i\pi\gamma\theta/L} \right|\\
      &\leq\frac1{(2L)^{d/2}}\sum_{e^{k_\varepsilon}\varepsilon L \leq
        |\gamma|}
      \left[\left|\frac{\pi\gamma}{L}\right|^{r/2}\left|(\hat
          v)_\gamma\right|\right] \left|\frac{\pi\gamma}{L}\right|^{-r/2}\\
      &\leq \frac1{(2L)^{d/2}}\left(\sum_{e^{k-1}\varepsilon L \leq
          |\gamma|< e^k\varepsilon L}
        \left|\frac{\pi\gamma}{L}\right|^{r} |(\hat
        v)_\gamma|^2\right)^{1/2} \left(\sum_{e^{k-1}\varepsilon L
          \leq |\gamma|}\left|\frac{\pi\gamma}{L}\right|^{-r}
      \right)^{1/2}\\
      &\leq C \left(\sum_{e^{k-1}\varepsilon L \leq |\gamma|<
          e^k\varepsilon L} \left|\frac{\pi\gamma}{L}\right|^{r-2/q}
        |(\hat v)_\gamma|^{2/p}\cdot
        \left|\frac{\pi\gamma}{L}\right|^{2/q}
        |(\hat v)_\gamma|^{2/q}\right)^{1/2}\\
      &\leq C \|v\|^{1/p}_{H^{rp-2p/q}}\cdot\langle
      -\Delta^P_L v,v\rangle^{1/(2q)}\\
      &\leq C \|v\|^{\eta}_{H^n}\cdot\epsilon^{1-\eta}
    \end{split}
  \end{equation*}
  if $p=\eta$, $q=1-\eta$ and $r=(n-1)\eta+2>d$ as
  $n>(d-2)\eta^{-1}+1$. One then completes the proof of
  Lemma~\ref{le:3} in the same way as that of Lemma~\ref{le:1}.
\end{proof}
\begin{proof}[Proof of Proposition~\ref{Lemma2}]
  \noindent %
  \noindent In the discrete case Proposition~\ref{Lemma2} is a
  consequence of Lemma~\ref{Lemma1} and Lemma~\ref{le:1} with
  $\varepsilon=C (\log L)^{-1/d}$.\\[2mm]
  To be able to apply Lemma~\ref{le:3} to $\varphi_0(\omega,L)$ in the
  continuous case, we need to show that, for any $n>d$,
  $\varphi_0(\omega,L)\in H^n$ with a bounded depending only on $n$
  not on $L$ or $\omega$. Therefore we use the first assumption on the random
  field $V_\omega$ i.e. that, for any $\alpha\in\N^d$,
  $\|\partial^\alpha V_\omega\|_{\omega,x,\infty}<+\infty$. Hence, as
  $\varphi_0(\omega,L)$ is an eigenvector of $-\Delta^P_L+V_\omega$,
  using the eigenvalue equation
  \begin{equation*}
    -\Delta_L^P\varphi_0(\omega,L)=(E_0^P[\omega,L]-V_\omega)\varphi_0(\omega,L) 
  \end{equation*}
  inductively, we see that
  \begin{equation*}
    \|\|\varphi_0(\omega,L)\|_{H^n}\|_{\omega,\infty}<+\infty. 
  \end{equation*}
  Proposition~\ref{Lemma2} in the continuous case is then a
  consequence of Lemma~\ref{Lemma1} and Lemma~\ref{le:3} with
  $\varepsilon=C (\log L)^{-1/d}$. This completes the proof of
  Proposition~\ref{Lemma2}.
\end{proof}
\section{The spectral gap of the random Hamiltonian}
\label{sec:dist-betw-first}
\noindent The main result of the present section is
\begin{proposition}
  \label{le:2}
  Let the first two eigenvalues of ${\rm H}^P_{\omega,L}$ be denoted
  by ${\rm E_0^P[\omega,L]}<{\rm E_1^P[\omega,L]}$. Then, for $p>0$, there
  exists $C>0$ such that, for $L$ sufficiently large and
  $\eta\in(0,1)$, one has
  \begin{equation*}
    \bP\left[ {\rm E_1^P[\omega,L]}-{\rm E_0^P[\omega,L]} \leq \eta
      L^{-d}\right]\leq C\eta\left[1+(\log L)^{d-2/d+\epsilon}\right]+L^{-p} 
  \end{equation*}
  with $\epsilon=0$ in the discrete setting resp.  $\epsilon>0$
  arbitrary in the continuous case.
\end{proposition}
\noindent In the localization regime, both the level-spacing and the
localization centers spacing have been studied in
e.g.~\cite{GeKl,KN}. The main difficulty arising in the present
setting is that the interval over which we need to control the spacing
is of length $C(\log L)^{-2/d}$; it is large compared to the length
scales dealt with in~\cite{GeKl,KN}.\\
\noindent Our analysis of the spectral gap relies on the description
of the ground state resulting from the analysis of the Anderson model
$H_\omega$ in the localized regime (see e.g.~\cite{Kirsch},
~\cite{Stol}). Under the assumptions made above on $H_\omega$, there
exists $I$ a compact interval containing $0$ such that, in $I$, the
assumptions of the Aizenman-Molchanov technique (see
e.g.~\cite{AENSS,MR2002h:82051}) or of the multi-scale analysis (see
e.g.\cite{GKsudec}) are satisfied. One proves
\begin{lemma}[\cite{GKsudec,Klopp2}]
  \label{Lemma4}
  There exists $\alpha>0$ such that, for any $p>0$, there exists $q>0$
  such that, for any $L\geq1$ and $\xi\in(0,1)$, there exists
  $\Omega_{I,\delta,L}\subset\Omega$ such that
  \begin{itemize}
  \item $\bP[\Omega_{I,\delta,L}]\geq 1-L^{-p}$,
  \item for $\omega\in\Omega_{I,\delta,L}$, one has that, if
    $\varphi_{n,\omega}$ is a normalized eigenvector of ${\rm
      H}_{\omega}|_{\Lambda_L}$ associated to $E_{n,\omega}\in I$, and
    $x_n(\omega)\in \Lambda_L$ is a maximum of
    $x\mapsto|\varphi_{n,\omega}(x)|$ on $\Lambda_L$ then, for
    $x\in\Lambda_L$, one has,
    \begin{equation}
      \label{eq:3}
      |\varphi_{n,\omega}(x)|\leq L^q\cdot
      \begin{cases}
        e^{-\alpha|x-x_n(\omega)|}\text{ in the discrete case,}\\
        e^{-\alpha|x-x_n(\omega)|^\xi}\text{ in the continuous case.}
      \end{cases}
    \end{equation}
  \end{itemize}
\end{lemma}
\noindent Note that, for a given eigenfunction, the maximum of its
modulus need not be unique but two maxima can not be further apart
from each other than a distance of order $\log L$. So for each
eigenfunction, we can choose a maximum of its modulus that we dub
center
of localization for this eigenfunction.\\
To prove Proposition~\ref{le:2}, we will   distinguish two cases
whether the localization centers associated to ${\rm E_0^P[\omega,L]}$
and ${\rm E_1^P[\omega,L]}$, say, respectively $x_0(\omega)$ and
$x_1(\omega)$ are close to or far away from each other.\\
In Lemma~\ref{Lemma3}, we show that the centers of localization being
close is a very rare event as a consequence of the Minami estimate.\\
In Lemma~\ref{Lemma5}, we estimate the probability of ${\rm
  E_0[\omega,L]}$ and ${\rm E_1[\omega,L]}$ being close to each other
when $x_0(\omega)$ and $x_1(\omega)$ are far away from each other. In
this case, ${\rm E_0[\omega,L]}$ and ${\rm E_1[\omega,L]}$ are
essentially independent of each other, and the estimate is obtained
using Wegner's estimate.
\begin{lemma}
  \label{Lemma3}
  For $p>0$, there exists $L_0>0$ such that, for $\lambda>0$, $L\geq
  L_0$ and $\eta\in(0,1)$, one has
  \begin{equation*}
    \bP\left[  
      \begin{aligned}
        &{\rm E_1^P[\omega,L]}-{\rm E_0^P[\omega,L]} \leq \eta\, L^{-d},\\
        &|x_0(\omega)-x_1(\omega)|\leq \lambda(\log L)^{1/\xi}
      \end{aligned}
    \right]\, \leq C\eta(\log L)^{d/\xi-2/d} + L^{-p} 
  \end{equation*}
  with $\xi=1$ in the discrete setting resp.  $\xi>1$ arbitrary in the
  continuous case.
\end{lemma}
\begin{proof}
  Let us start with the discrete setting.
  Fix $p>0$ and let $q$ be given by Lemma~\ref{Lemma4}.
  The basic observation  following from Lemma~\ref{Lemma4} is that,
  for $\omega\in\Omega_{I,\delta,L}$, if $x_n(\omega)$ is the
  localization center of $\varphi_n(\omega,L)$ and $l\leq L$, then
  \begin{equation}
    \label{eq:10}
    \|(H^0_\omega-E_n^P(\omega,L))\tilde\varphi_n(\omega,L,l)\|+
    \left|\|\tilde\varphi_n(\omega,L,l)\|-1\right|\leq C L^qe^{-\alpha l}.
  \end{equation}
  where
  \begin{itemize}
  \item $H^0_\omega=[H^P_{\omega,L}]_{|x_n(\omega)+\Lambda_l}$ is
    $H^P_{\omega,L}$ restricted to the cube $x_n(\omega)+\Lambda_l$,
  \item $\tilde\varphi_n(\omega,L,l)=\car_{x_n(\omega)+\Lambda_l}
    \varphi_n(\omega,L)$ is the eigenfunction $\varphi_n(\omega,L)$
    restricted to the cube $x_n(\omega)+\Lambda_l$.
  \end{itemize}
  To apply the observation above we pick a covering $(C_j)_{0\leq
  j\leq J}$ of $\Lambda_L$ by cubes of side length of order $\log L$ i.e.
  $\Lambda_L\subset\bigcup_{0\leq j\leq J}C_j$. Then the number of cubes $J$
  can be estimated by $J\leq C L^d(\log L)^{-d}$ and there exists $C>0$ 
  (depending on $\lambda$, $q$ and $\nu$)
  such that, if $|x_0(\omega)-x_1(\omega)|\leq \lambda\log L$ and
  $l\geq C\lambda\log L$, there exists a cube $C_j$ (containing
  $x_0(\omega)$) such that, for $L$ sufficiently large
  \begin{multline*}
    \sum_{k=0}^1\left(\|(H^j_\omega-E_k^P(\omega,L))
      \tilde\varphi_k(\omega,L,j)\|+|\|\tilde\varphi_k(\omega,L,j)\|-1|\right)
    \\+|\langle\tilde\varphi_0(\omega,L,j),
    \tilde\varphi_1(\omega,L,j)\rangle|\leq L^{-\nu}/2
  \end{multline*}
  where we have set $q-C\lambda\alpha<-\nu$ (see~\eqref{eq:10}) and
  \begin{itemize}
  \item $H^j_\omega$ is the operator $H_\omega$ restricted to the cube
    $C_j+\Lambda_l$,
  \item $\tilde\varphi_k(\omega,L,j)=\car_{C_j+\Lambda_l}
    \varphi_k(\omega,L)$ for $k\in\{0,1\}$.
  \end{itemize}
  $\;$ \\
  Let $C$ be given by Lemma~\ref{Lemma1} and define $I=[0,2C(\log
  L)^{-2/d}]$. Decompose $I\subset\cup_{m=0}^{2M+1}I_m$ where
  \begin{itemize}
  \item $I_m$ are intervals of length $4\eta L^{-d}$,
  \item for $m\in\{0,\dots,M-1\}$, $I_{2m}\cap I_{2(m+1)}
    =\emptyset=I_{2m+1}\cap I_{2m+3}$,
  \item for $m\in\{0,\dots,M\}$, $I_{2m}\cap I_{2m+1}$ is of length
    $2\eta L^{-d}$.
  \end{itemize}
  One can choose $M\leq CL^d(\log L)^{-2/d}\eta^{-1}$.  This
  implies that, for $L$ sufficiently large,
  \begin{equation*}
    \left\{\omega;\ 
      \begin{aligned}
        &{\rm E_1^P[\omega,L]}-{\rm E_0^P[\omega,L]} \leq \eta\, L^{-d}\\
        &|x_0(\omega)-x_1(\omega)|\leq \lambda\log L
      \end{aligned}
    \right\}\subset\Omega_1\cup\Omega_2
  \end{equation*}
  where $\Omega_1=\Omega\setminus\Omega_{I,\delta,L}$ and
  \begin{equation*}
    \Omega_2=\bigcup_{j=1}^J\bigcup_{m=0}^{2M+1}\{ 
    (H_\omega)_{|C_j+\Lambda_l}\text{ has two eigenvalues in }I_m\}.
  \end{equation*}
  By Lemma~\ref{Lemma4}, we know that
  \begin{equation*}
    \pro[\Omega_1]\leq L^{-p}
  \end{equation*}
  Minami's estimate (H.2) and the estimate on $M$ tells us that
  \begin{equation*}
    \pro[\Omega_2]\leq CL^{2d}(\log L)^{-d-2/d}\eta^{-1}
    (\eta L^{-d}(C\log L)^d)^2\leq C\eta(\log L)^{d-2/d}.
  \end{equation*}
  This completes the proof for the discrete setting.  The proof for
  the continuous case is very similar. One has to replace
  $\car_{x_n(\omega)+\Lambda_l}$ by a smooth version of the
  characteristic function of the cube $x_n(\omega)+\Lambda_l$ (see for
  example \cite{Stosto}), resp. change the length scale $\log L$ to
  $(\log L)^{1/\xi}$ in the side length of the boxes where one
  restricts the eigenfunctions. This is necessary because of the
  weaker estimate in Lemma~\ref{Lemma4}.
  This completes the proof of Lemma~\ref{Lemma3}.
\end{proof}
\noindent We now estimate the probability of the spectral gap being
small conditioned on the fact that the localization centers are far
away from one another. We prove
\begin{lemma}
  \label{Lemma5}
  For any $p>0$, there exists $\lambda>0$ and $ C>0$ such that, for
  $L$ sufficiently large and $\eta\in(0,1)$, one has
  \begin{equation*}
    \bP\left[
      \begin{aligned}
        &{\rm E_1^P[\omega,L]}-{\rm E_0^P[\omega,L]} \leq \eta\, L^{-d},\\
        &|x_0(\omega)-x_1(\omega)|\geq \lambda(\log L)^{1/\xi}
      \end{aligned}
    \right]\,
    \leq C\eta+L^{-p}
  \end{equation*}
  with $\xi=1$ in the discrete setting, resp. $\xi>1$ arbitrary in
  the continuous case.
\end{lemma}
\begin{proof}
  Using the same line of reasoning as in the proof of
  Lemma~\ref{Lemma3} we give the proof in the discrete setting.\\[2mm]
  Fix $\nu>2d+p$ and split the interval $[0,C(\log L)^{-2/d}]$ into
  intervals of length $L^{-\nu}$ as in the proof of
  Lemma~\ref{Lemma3}. By Minami's estimate, we know that, for $L$
  sufficiently large
  \begin{equation}
    \label{eq:11}
    \bP\left[ 
      {\rm E_1^P[\omega,L]}-{\rm E_0^P[\omega,L]} \leq L^{-\nu}
    \right]\,\leq C(\log L)^{-2/d}L^\nu L^{2(d-\nu)}\leq L^{-p}.
  \end{equation}
  So we may assume that ${\rm E_1^P[\omega,L]}-{\rm E_0^P[\omega,L]} \geq
  L^{-\nu}$.\\[2mm]
  As in the proof of Lemma~\ref{Lemma3}, pick a covering of
  $\Lambda_L$ by cubes, say $(C_j)_{0\leq j\leq J}$ of side length
  less than $\log L$ such that $J$, the number of cubes, satisfies
  $J\leq C(L/\log L)^d$.\\
  Assume that $C_j$ is the cube containing
  $x_1(\omega)$, ${\rm E_1[\omega,L]}-{\rm E_0[\omega,L]} \leq \eta\,
  L^{-d}$ and $|x_0(\omega)-x_1(\omega)|\geq \lambda\log L$.
  Let $\Lambda_j^c=\Lambda_L\setminus (C_j+\Lambda_{3/4\lambda\log
   L }) $. Define the operators $(H_\omega)_{|\Lambda_j^c}$,
  resp. $(H_\omega)_{|C_j+\Lambda_{\lambda\log L/4}}$  to be the
  restriction of $H^P_{\omega,L}$ to $\Lambda_j^c$,
  resp. $C_j+\Lambda_{\lambda\log L/4}$, with Dirichlet boundary
  conditions. If $\lambda\geq 8$ and $L$ is large enough, we know that
  \begin{itemize}
  \item dist$(x_0(\omega),\partial\Lambda_j^c)\geq\lambda \log
    L-3/4\lambda\log L -\log L\geq \lambda\log L/8 $,
  \item dist$(x_1(\omega),\partial(C_j+\Lambda_{\lambda\log L/4}))\geq
    \lambda\log L/4$
  \item dist$(\Lambda_j^c, C_j+\Lambda_{\lambda\log L/4}) \geq
    \lambda\log L/2\geq R $ 
  \end{itemize}
  with $R>0$ as in the decorrelation assumption (H0). Hence, for $\lambda$ sufficiently large, using the
  estimate~(\ref{eq:10}) for the operators $(H_\omega)_{|\Lambda_j^c}$
  and $(H_\omega)_{|C_j+\Lambda_{\lambda\log L/4}}$, we know that:
  \begin{itemize}
  \item The operator $(H_\omega)_{|C_j+\Lambda_{\lambda\log L/4}}$
    admits an eigenvalue, say $\tilde E_1(\omega)$, that satisfies
    $|\tilde E_1(\omega) - E_1^P(\omega)|\leq L^{-2\nu}$;
  \item The operator $(H_\omega)_{|\Lambda_j^c}$ admits an eigenvalue,
    say $\tilde E_0(\omega)$, that satisfies $|\tilde E_0(\omega) -
    E_0(\omega)|\leq L^{-2\nu}$. Moreover, as
    $(H_\omega)_{|\Lambda_L^c}$ is the Dirichlet restriction of
    $H^P_{\omega,L}$, its eigenvalues are larger than those of
    $H^P_{\omega,L}$. In particular, its second eigenvalue is larger
    than $E_1(\omega)$. Hence, up to a small loss in probability, we
    may assume it is larger than $E_0(\omega)+L^{-\nu}$ as we know the
    estimate~(\ref{eq:11}). This implies that we may assume that
    $\tilde E_0(\omega)$ is the ground state of
    $(H_\omega)_{\Lambda_j^c}$.
  \end{itemize}
  So we obtain  
  \begin{equation*}
    \left\{\omega;\ 
      \begin{aligned}
        {\rm E_1^P[\omega,L]}-{\rm E_0^P[\omega,L]}&\leq\eta\,L^{-d},\\
        |x_0(\omega)-x_1(\omega)|&\geq\lambda\log L
      \end{aligned}
    \right\}\subset\Omega_1
    \cup\bigcup_{1\leq j\leq J}\Omega_j
  \end{equation*}
  where $\Omega_1=\Omega\setminus(\Omega_{I,\delta,L}\cup\{\omega ;\
  {\rm E_1[\omega,L]}-{\rm E_0[\omega,L]} \leq L^{-\nu}\}$ and
  \begin{equation*}
    \Omega_j=\left\{\omega;
      \text{dist}(\sigma((H_\omega)_{|C_j+\Lambda_{\lambda\log L/4}}), 
      \inf\sigma((H_\omega)_{|\Lambda_j^c}))\leq \eta L^{-d}+L^{-\nu}
    \right\}
  \end{equation*}
  As $(H_\omega)_{|C_j+\Lambda_{\lambda\log L/4}}$ and
  $(H_\omega)_{|\Lambda_j^c}$ are independent of each other, we
  estimate the probability of $\Omega_j$ using Wegner's estimate to
  obtain
  \begin{equation*}
    \bP[\Omega_j]\leq C(\eta L^{-d}+L^{-\nu})(\log L)^d.
  \end{equation*}
  Hence, one obtains
  \begin{equation*}
    \begin{split}
      \bP\left[
        \begin{aligned}
          &{\rm E_1[\omega,L]}-{\rm E_0[\omega,L]} \leq \eta\, L^{-d},\\
          &|x_0(\omega)-x_1(\omega)|\geq \lambda\log L
        \end{aligned}
      \right] &\leq C(\eta L^{-d}+L^{-\nu})(\log L)^d\frac{L^d}{(\log
        L)^d}+2L^{-p} \\&\leq C(\eta+L^{-p})
    \end{split}
  \end{equation*}
  if $\nu>d+p$.\\[2mm]
  This completes the proof in the discrete setting.  To prove
  Lemma~\ref{Lemma5} for the continuous case, one does the same
  modifications as in the proof of Lemma~\ref{Lemma3} in the
  continuous setting.
\end{proof}
\noindent Setting $\varepsilon=d(1/\xi-1)$, Proposition~\ref{le:2}
then follows from Lemma~\ref{Lemma3} and Lemma~\ref{Lemma5}.
\section{Proof of Theorem~\ref{Theorem1}}
\label{sec:proof-cond-single}
\noindent Defining $\pi_0=|\varphi_0\rangle\langle \varphi_0|$ and
applying the definition of the ground state, we can estimate
\begin{align*}
  {\rm E_1^P[\omega,L]}&\|(1-\pi_0)\varphi^{\rm GP}\| +{\rm
    E_0^P[\omega,L]}\| \pi_0 \varphi^{\rm GP}\|\\ &\leq {\rm
    E}_{\omega,L}^{\rm GP}\|(1-\pi_0)\varphi^{\rm GP}\| +{\rm
    E}_{\omega,L}^{\rm GP}\| \pi_0 \varphi^{\rm GP}\|,
\end{align*}
respectively
\begin{equation*}
  \left( {\rm E_1^P[\omega,L]}-{\rm E}_{\omega,L}^{\rm GP}\right)
  \|(1-\pi_0)\varphi^{\rm GP}\| \leq \left( {\rm E}_{\omega,L}^{\rm
      GP}-{\rm E_0^P[\omega,L]}\right) \| \pi_0 \varphi^{\rm GP}\|.
\end{equation*}
As a consequence of Proposition~\ref{Lemma2} and
Proposition~\ref{le:2}, we know with a probability larger than
$1-(C\eta+L^{-p})$ that, for $\eta\in(0,1)$ and $f_d$ defined
in~(\ref{eq:12}) the estimates
\begin{equation*}
  {\rm E_1^P[\omega,L]}-{\rm E}_{\omega,L}^{\rm GP} \geq {\rm
    E_1^P[\omega,L]}-{\rm E_0^P[\omega,L]} \geq \eta L^{-d}
  [1+(\log L)^{d-2/d+\epsilon}]^{-1} 
\end{equation*}
and
\begin{equation*}
  {\rm E}_{\omega,L}^{\rm GP}-{\rm E_0^P[\omega,L]}\leq C U f_d(\log L)
\end{equation*}
are satisfied.
We obtain
\begin{equation*}
  \|(1-\pi_0)\varphi^{\rm GP}\| \leq C U f_d(\log L) \eta^{-1} L^d[1+(\log L)^{d-2/d+\epsilon}]\;
  \|\pi_0\varphi^{\rm GP}\|
\end{equation*}
and
\begin{equation*}
  \begin{split}
    |\langle\varphi_0,\varphi^{\rm GP}\rangle|^2&=\| \pi_0
    \varphi^{\rm GP}\|^2= 1-\| (1-\pi_0)\varphi^{\rm GP}\|^2\\
    &\geq 1-\left[CU f_d(\log L) \eta^{-1} L^d [1+(\log
      L)^{d-2/d+\epsilon}]\right]^2.
  \end{split}
\end{equation*}
Applying the assumption concerning the coupling constant $U$ i.e.
\begin{equation*}
  U=U(L)=  o\left(L^{-d}[1+(\log L)^{d-2/d+\epsilon}]^{-1}[f_d(\log
    L)]^{-1}\right)    
\end{equation*}
and setting
\begin{equation*}
  \eta=\eta(L)=\sqrt{|U(L)
    L^{d}[1+(\log L)^{d-2/d+\epsilon}]f_d(\log L)|}
\end{equation*}
we get that, $\eta(L)\to0$ when $L\to+\infty$ and for some $C>0$,
\begin{equation*}
  \bP(\{\omega; ||\langle\varphi_0,\varphi^{\rm GP}\rangle|-1|\geq
  C\eta(L)\})\leq C(\eta(L)+L^{-p}).
\end{equation*}
This completes the proof of Theorem~\ref{Theorem1}.

\end{document}